\documentclass[12pt,reqno]{amsart}
\usepackage[T1]{fontenc}
\usepackage[english]{babel}
\usepackage[utf8]{inputenc}
\usepackage[square,numbers]{natbib}
\usepackage{amssymb}
\usepackage{amsmath}
\usepackage{xcolor}
\usepackage{amsfonts}
\usepackage{amsthm}
\usepackage{csquotes}
\usepackage[a4paper,margin=2.5cm]{geometry}
\usepackage{graphicx} 
\usepackage{comment}
\usepackage{enumitem}

\usepackage[colorlinks]{hyperref}
\hypersetup{
 colorlinks = true,
 linkcolor = blue,
}
\setcitestyle{numbers}
\theoremstyle{plain}

\newcommand{\esp}{\operatorname{\mathbb{E}}}

\newcommand{\var}{\operatorname{var}}
\newcommand{\cov}{\operatorname{cov}}
\newcommand{\norm}[1]{\|#1\|}
\newcommand{\real}{\mathbb{R}}
\newcommand{\complex}{\mathbb{C}}
\newcommand{\tr}{\operatorname{tr}}
\newcommand{\Id}{\operatorname{Id}}
\newcommand{\diag}{\operatorname{diag}}

\title[Spectrum of random quantum channels]{Limiting spectral distribution of random self-adjoint quantum channels}
\author[C.\ Lancien \and P.\ O.\ Santos \and P.\ Youssef]{%
        Cécilia Lancien \and 
        Patrick Oliveira Santos \and 
        Pierre Youssef
        }

\address{C\'ecilia Lancien. CNRS \& Institut Fourier UMR 5582, Universit\'e Grenoble Alpes, Grenoble, France.}
\email{cecilia.lancien@univ-grenoble-alpes.fr}
\address{Patrick Oliveira Santos. Laboratoire d'Analyse et de Math\'ematiques Appliqu\'ees UMR 8050, Universit\'e Gustave Eiffel, Universit\'e Paris Est Cr\'eteil, Marne-la-Vall\'ee, France.}
\email{patrick.oliveirasantos@u-pem.fr}
\address{Pierre Youssef. Division of Science, NYU Abu Dhabi, Abu Dhabi, UAE \& Courant Institute of Mathematical Sciences, New York University, New York, USA.}
\email{yp27@nyu.edu}

\begin{document}
\newtheorem{theorem}{Theorem}[section]

\newtheorem{corollary}[theorem]{Corollary}
\newtheorem{lemma}[theorem]{Lemma}
\newtheorem{conjecture}[theorem]{Conjecture}
\newtheorem{proposition}[theorem]{Proposition}

\theoremstyle{definition}
\newtheorem{example}[theorem]{Example}
\newtheorem{definition}[theorem]{Definition}
\newtheorem{remark}[theorem]{Remark}

\maketitle
\begin{abstract}
We study the limiting spectral distribution of quantum channels whose Kraus operators are sampled as $ n\times n$ random Hermitian matrices satisfying certain assumptions. 
We show that when the Kraus rank goes to infinity with $n$, the limiting spectral distribution (suitably rescaled) of the corresponding quantum channel coincides with the semi-circle distribution. When the Kraus rank is fixed, the limiting spectral distribution is no longer the semi-circle distribution. It corresponds to an explicit law, which can also be described using tools from free probability. 
\end{abstract}

\section{Introduction}\label{sec: introduction}


In quantum physics, the state of an $n$-dimensional system is described by a density operator on $\mathbb C^n$, i.e.~$\rho\in M_n(\mathbb C)$ a Hermitian positive semidefinite matrix with trace $1$: $\rho\succeq 0$ and $\tr(\rho)=1$. A transformation of such a quantum system is described by a quantum channel on $M_n(\complex)$, i.e.~$\Phi: M_n(\mathbb C)\to M_n(\mathbb C)$ a completely positive and trace-preserving linear map. We recall that a linear map $\Phi:M_n(\mathbb C)\to M_n(\mathbb C)$ is said to be
 \begin{itemize}
	\item positive if it preserves the fact of being Hermitian and positive semidefinite: for all $X\in M_n(\mathbb C)$, $X\succeq 0$ implies that $\Phi(X)\succeq 0$;
        \item completely positive if $\Phi\otimes \Id:M_{n^2}(\mathbb C)\to M_{n^2}(\mathbb C)$ is positive (where $\Id:M_n(\mathbb C)\to M_n(\mathbb C)$ denotes the identity map);
	\item trace-preserving if it preserves the trace: for all $X\in M_n(\mathbb C)$, $\tr(\Phi(X))=\tr(X)$.
\end{itemize}
A quantum channel $\Phi$ thus maps quantum states to quantum states (and so does $\Phi\otimes\mathrm{id}$).
  	
The action of a completely positive map $\Phi$ on $M_n(\mathbb C)$ can always be described in the following (non-unique) way, called a Kraus representation of $\Phi$ (see e.g.~\cite[Section 2.3.2]{aubrun2017} or \cite[Chapter 2]{wolf2012}): There exist $d\in\mathbb N$ and $K_1,\ldots, K_d\in M_n(\mathbb C)$, called Kraus operators of $\Phi$, such that
\begin{equation} \label{eq:Kraus}
	\Phi:X\in M_n(\mathbb C) \mapsto \sum_{i\in[d]} K_iXK_i^*\in M_n(\mathbb C), 
\end{equation}
where $K_i^*$ stands for the adjoint of $K_i$ (and where we have used the shorthand notation $[d]$ for the set of integers $\{1,\ldots,d\}$). The fact that $\Phi$ is trace-preserving is equivalent to the following constraint on the Kraus operators $K_1,\ldots, K_d$:
\[ \sum_{i\in[d]} K_i^*K_i = \Id. \]
The smallest $d$ such that an expression of the form of equation \eqref{eq:Kraus} for $\Phi$ exists is called the Kraus rank of $\Phi$. It is always at most $n^2$ for a completely positive map on $M_n(\mathbb C)$.

Given a completely positive map $\Phi$ on $M_n(\complex)$, its adjoint (or dual) is defined by duality with respect to the Hilbert-Schmidt inner product, i.e.~as the map $\Phi^*$ on $M_n(\complex)$ (which is completely positive as well) such that, for all $X, Y\in M_n(\complex)$,
\begin{equation} \label{eq:def-adjoint}
\tr\left(X\Phi^*(Y)\right)=\tr\left(\Phi(X)Y\right). 
\end{equation}
$\Phi$ being trace-preserving is equivalent to $\Phi^*$ being unital, i.e.~such that $\Phi^*(\Id)=\Id$.

Note that, identifying $M_n(\mathbb C)$ with $\mathbb C^n\otimes\mathbb C^n$, a linear map $\Phi:M_n(\mathbb C)\to M_n(\mathbb C)$ can equivalently be seen as a linear map $M_\Phi:\mathbb C^n\otimes\mathbb C^n\to\mathbb C^n\otimes\mathbb C^n$, i.e.~an element of $M_{n^2}(\mathbb C)$. Concretely, a completely positive linear map
\[ \Phi:X\in M_n(\mathbb C) \mapsto \sum_{i\in[d]} K_iXK_i^*\in M_n(\mathbb C) \]
can be identified with
\begin{equation} \label{eq:def-M-Phi} 
M_\Phi = \sum_{i\in[d]} K_i\otimes\overline K_i \in M_{n^2}(\mathbb C),
\end{equation}
where $\overline K_i$ stands for the entry-wise conjugate of $K_i$, in the canonical basis of $\complex^n$. This identification preserves the spectrum, i.e.~$\mathrm{spec}(\Phi)=\mathrm{spec}(M_\Phi)$. Moreover, the matrix version of the adjoint $\Phi^*$ of $\Phi$ is simply the adjoint of $M_\Phi$, i.e.~$M_{\Phi^*}=M_\Phi^*$.

In this paper, we will consider quantum channels $\Phi$ whose Kraus operators are Hermitian matrices, which ensures that $\Phi$ itself is Hermitian (in the sense that it is equal to its adjoint $\Phi^*$, as defined by equation \eqref{eq:def-adjoint}), or equivalently that its matrix version $M_{\Phi}$ is Hermitian. This may look like a restrictive setting, but it, in fact, encompasses all Hermitian quantum channels. Indeed, it is not hard to see that any Hermitian completely positive map $\Phi$ on $M_n(\complex)$ admits a Kraus representation with Hermitian Kraus operators, where we can additionally guarantee that the number of such operators is at most twice the Kraus rank of $\Phi$ (i.e.~in general at most $2n^2$). Concretely, if $K_1,\ldots,K_d\in M_n(\complex)$ are $d$ Kraus operators for $\Phi$, then $K_{1,R},K_{1,I},\ldots,K_{d,R},K_{d,I}\in M_n(\complex)$ are $2d$ Hermitian Kraus operators for $\Phi$, where given $K\in M_n(\complex)$, we set $K_R=(K+K^*)/2,\,K_I=-i(K-K^*)/2$, so that $K=K_R+iK_I$ and $K_R^*=K_R,\,K_I^*=K_I$. The latter claim follows from the observation that,if $M_\Phi^*=M_\Phi$, then we can re-write $M_\Phi=(M_\Phi+M_\Phi^*)/2$, and thus
\[ M_\Phi = \frac{1}{2}\sum_{j\in[d]} \left(K_j\otimes\overline{K}_j+K_j^*\otimes\overline{K}_j^*\right) = \sum_{j\in[d]}\left(K_{j,R}\otimes\overline{K}_{j,R}+K_{j,I}\otimes\overline{K}_{j,I}\right). \]
Now, many interesting and well-studied classes of quantum channels are Hermitian: depolarizing channels, dephasing channels, and Pauli channels, to name just a few. Note that such channels are, in particular, unital.

This work aims to study the spectrum of a randomly generated quantum channel as the underlying dimension $n$ goes to infinity. More precisely, given a quantum channel  $\Phi$ on $M_n(\mathbb C)$ whose Kraus operators $K_1,\ldots, K_d\in M_n(\complex)$ have been randomly sampled, we aim to characterize its asymptotic spectrum (i.e.~its spectrum in the limit where $n$ goes to infinity). 
Our study will involve two regimes, one where the Kraus rank $d$ is fixed and only $n$ grows, and the other one where $d=d(n)$ is a growing function of $n$. 

Previous related works were mostly concerned with identifying the spectral gap of a random quantum channel, i.e.~the difference between its largest and second largest eigenvalues. Indeed, it is known that a quantum channel $\Phi$ on $M_n(\mathbb C)$ always has its largest (in modulus) eigenvalue $\lambda_1(\Phi)$ equal to $1$ (with an associated eigenvector which is a positive semidefinite matrix), implying that $\Phi$ always has a fixed state. It was established that random quantum channels generically have their largest eigenvalue $1$ isolated from the rest of the spectrum. All other eigenvalues are of order at most $1/\sqrt{d}$ (in modulus). This was shown first for specific models, where Kraus operators were sampled either as independent Haar unitaries \cite{hastings2007,pisier2014}, or as blocks of a Haar isometry \cite{gonzalez2018}, or as independent Ginibre matrices (i.e.~matrices having i.i.d.~Gaussian entries) \cite{lancien2022}, and recently in greater generality \cite{lancien2023note}. 

On the other hand, much less is known concerning the asymptotic distribution of the bulk of the spectrum. To understand this, one studies the empirical spectral distribution, which, given a matrix $M \in M_n(\complex)$, is defined as 
\begin{align*}
    \mu_{M}:=\frac{1}{n}\sum_{k\in[n]} \delta_{\lambda_k(M)},
\end{align*}
where $\lambda_k(M)$ is the $k$th eigenvalue of $M$. 
In \cite{bruzda2009}, a model where Kraus operators are sampled as blocks of a Haar isometry was studied. It was conjectured, from heuristic arguments and numerical simulations, that in the regime where $d=d(n)=n^2$, the empirical spectral distribution (rescaled by a factor $\sqrt{d}$) of such random quantum channel converges towards a circular distribution as $n$ grows. It was later shown in \cite{aubrun2012} that, in the same regime $d=d(n)$ of order $n^2$, the empirical singular value distribution (again rescaled by a factor $\sqrt{d}$) of a random quantum channel whose Kraus operators are sampled as independent Ginibre matrices converges towards a quarter-circular distribution as $n$ grows. 
We believe that for both models (and for more general ones) and in any regime $d=d(n)\to \infty$, the limiting spectral distribution coincides with the circular law.

In this work, we embark on the investigation of the limiting spectral distribution of quantum channels in the Hermitian setting. This serves as a precursor to our broader research goals involving the non-Hermitian case and the conjecture presented above. Classical problems in Random Matrix Theory have traditionally prioritized the exploration of Hermitian matrices before delving into their non-Hermitian counterparts, partly due to the inherent technical complexities introduced by non-Hermitian systems, notably the instability of the spectrum under perturbations (see \cite[Chapter 11]{bai2010spectral}). For instance, while the limiting spectral distribution of an $n\times n$ Hermitian matrix with i.i.d.~(up to symmetry) centered entries of variance $1/n$ was shown to be Wigner's semi-circle distribution, it is almost half a century later that the analogous result in the non-Hermitian case, Girko's circular law theorem, was established in full generality \cite{tao2010random}. Understanding the limiting spectral distribution in the non-Hermitian case follows the Hermitization technique invented by Girko \cite{girko1983circular,girko1994circular}. The latter requires a quantitative control on the smallest singular value of the corresponding random matrix model (see  \cite{bordenave2012around} for an introduction to the method), making the problem significantly more involved. In the context of quantum channels, the associated random matrix model exhibits dependencies among its entries, which adds a layer of difficulty compared to classical random matrix models, even in the Hermitian setting.

Now, let us shift our focus back to the specific context of this paper. Here, we will consider the case where the Kraus operators of the quantum channel $\Phi$ are chosen to be random Hermitian operators. This ensures that the resulting completely positive map $\Phi$, or equivalently its matrix version $M_{\Phi}$, is Hermitian. Conversely, as explained earlier, any Hermitian completely positive map $\Phi$ can be written with Hermitian Kraus operators. Moreover, we aim to keep our assumptions regarding the distribution of these random Kraus operators as minimal as possible. Our objective is to gain a comprehensive understanding of the spectrum of $M_{\Phi}$ as the dimension $n$ increases without imposing specific constraints on the scaling relationship between $d$ and $n$.

To keep the introduction light, we state an informal version of our main result and differ from the rigorous statement, requiring notions from free probability, to Section \ref{sec: main results} (see Theorem \ref{theorem: main theorem}). 

\begin{theorem}[Informal Statement] \label{th:implication}
    Let $W_1,\ldots,W_d\in M_n(\complex)$ be independent Hermitian random matrices such that, for each $i\in[d]$, $\esp(W_i)=0$ and $\esp(W_i^2)=\Id$. Suppose additionally that, for each $i\in[d]$, the spectral distribution of $W_i$ converges to some distribution $\mu_i$ as $n\to\infty$ and that the joint asymptotic distribution of the family $(W_i)_{i\in[d]}$ is determined by the family $(\mu_i)_{i\in[d]}$. Next, set $K_i=W_i/\sqrt{d}$ for each $i\in[d]$ and define $\Phi$ as the random completely positive map on $M_n(\complex)$ having $K_1,\ldots,K_d$ as Kraus operators, i.e.
    \[ \Phi:X\in M_n(\complex)\mapsto \sum_{i\in[d]} K_iXK_i\in M_n(\complex). \]
    First, $\Phi$ is, on average, trace-preserving and unital. And second,
    \begin{itemize}
        \item If $d$ is fixed, then the spectral distribution of $\Phi-\esp(\Phi)$ converges to a specific distribution, depending only on $\mu_1,\ldots,\mu_d$, as $n\to \infty$.
        \item If $d=d(n)\to\infty$ as $n\to\infty$ and $W_i$ are i.i.d, then the spectral distribution of $\sqrt{d}(\Phi-\esp(\Phi))$ converges to the semicircular distribution as $n\to\infty$.
    \end{itemize} 
\end{theorem}

The specific assumption underlying the family of independent matrices $(W_i)_{i\in[d]}$ is precisely asymptotic freeness, a concept introduced by Voiculescu \cite{voiculescu1991limit} within the framework of Free Probability. Asymptotic freeness proves instrumental in a systematical examination of the asymptotic behavior of random matrices. Importantly, this assumption is not restrictive as many classical independent random matrices are known to be asymptotically free. However, this property is no longer valid in the context of tensors (see \cite{collins2017freeness}), preventing a direct derivation of the asymptotic joint spectrum of these random tensors, which is the setting of interest in this paper. To overcome this, we analyze the moment method directly on our model in order to derive the limiting distribution. 

The required notions from Free Probability are presented in Section~\ref{sec: main results}, where the main rigorous theorem is stated. The proofs are carried out in Section~\ref{sec: proofs}. 
Section~\ref{sec: examples} regroups a variety of examples for which our result applies, while Section~\ref{sec: perspectives} offers perspectives for future work.

\subsection*{Acknowledgments} Part of this work was completed during a stay of the second named author at New York University in Abu Dhabi, partly funded by a doctoral mobility grant delivered by Universit\'e Gustave Eiffel; he would like to thank both institutions for their support and the excellent working assumptions. The first named author was supported by the ANR projects ESQuisses (grant number ANR-20-CE47-0014-01), STARS (grant number ANR-20-CE40-0008) and QTraj (grant number ANR-20-CE40-0024-01).


\section{Preliminaries and rigorous statement of the main result}\label{sec: main results}
\subsection{Preliminaries}\label{subsec: free probability}
We begin by recalling some notation from Free Probability; see \cite{nica2006lectures}. A noncommutative probability space is a pair $(\mathcal{A},\tau)$, where $\mathcal{A}$ is a unital algebra equipped with a tracial state $\tau$, that is, $\tau$ is linear, $\tau(1)=1$ and $\tau(ab)=\tau(ba)$. We say that subalgebras $\mathcal{A}_1,\ldots, \mathcal{A}_d$ are free if 
\begin{align}\label{freeness identity}
    \tau(a_1\ldots a_p)=0
\end{align}
whenever $a_i \in \mathcal{A}_{j_i}$, $\tau(a_i)=0$ for all $i\in[p]$ and $j_1 \ne j_2 \ne \cdots \ne j_p$. We say that random variables $a_1,\ldots,a_d \in \mathcal{A}$ are free if their algebras are free. The distribution of a self-adjoint variable $a=a^*$ is defined as the collection of moments
\begin{align*}
    \{\tau(a^p):\, p \in \mathbb{N}\}.
\end{align*}
In particular, there always exists a measure $\mu$ on $\real$ such that
\begin{align*}
    \tau(a^p)=\int x^p\, \text{d}\mu,
\end{align*}
for every $p\in\mathbb N$. Such a measure is also called the distribution of $a$. A particular example of free variables is the free semicircular system $(s_1,\ldots, s_d)$, whose moments satisfy
\begin{align*}
    \tau(s_{i_1}\cdots s_{i_p})=\sum_{\pi \in NC_2(p)} \prod_{(l,k) \in \pi} \delta_{i_li_k}.
\end{align*}
Here, $NC_2(p)$ is the set of all noncrossing pair partitions of $[p]$, namely, each block of $\pi \in NC_2(p)$ has cardinality $2$ and $\pi$ does not contain two blocks $\{i,j\},\{k,l\}$ such that $i<k<j<l$. In particular, the moments of the semicircular variable are
\begin{align*}
    \tau(s^p)=|NC_2(p)|=\int x^p f_{sc}(x)\, \text{d}x,
\end{align*}
where $f_{sc}$ denotes the density of the semi-circle distribution $\mu_{sc}$ 
\begin{align}\label{semicircular density}
    f_{sc}(x)=\frac{1}{2\pi}\sqrt{4-x^2} \mathbf{1}_{|x| \le 2}.
\end{align}
Given a matrix $M \in M_n(\complex)$, we define its normalized trace as
\begin{align*}
    \tau^{(n)}(M):=\frac{1}{n}\tr(M).
\end{align*}
In particular, for two matrices $M,N \in M_n(\complex)$, we have
\begin{align*}
    \tau^{(n^2)}(M\otimes N)=\tau^{(n)}\otimes \tau^{(n)}(M \otimes N)=\tau^{(n)}(M)\tau^{(n)}(N).
\end{align*}
We say that a random Hermitian matrix $M\in M_n(\complex)$ is \textit{normalized} if
\begin{align*}
    &\esp\left(\tau^{(n)}(M)\right)=0;\\
    &\esp \left(\tau^{(n)}(M^2)\right)=1.
\end{align*}

We say that a sequence $(M_n)_{n\in\mathbb N}$ of random Hermitian matrices $M_n\in M_n(\complex)$ such that $\tr(M_n^p)$ is integrable for all $p \ge 1$ converges weakly in probability (resp.~in expectation) to $\mu$, if $(\mu_{M_n})_{n\in\mathbb N}$ converges weakly to $\mu$ in probability (resp.~in expectation). Equivalently, for any $p \ge 1$, we have
\begin{align*}
    \tau^{(n)}(M_n^p)=\frac{1}{n}\text{tr}(M_n^p) \underset{n\to\infty}{\to} \int x^p\, \text{d}\mu\
\end{align*}
in probability and
\begin{align*}
    \esp \left(\tau^{(n)}(M_n^p)\right)=\esp\left(\frac{1}{n}\text{tr}(M_n^p) \right)\underset{n\to\infty}{\to} \int x^p\, \text{d}\mu
\end{align*}
for the convergence in expectation. Whenever $(\mu_{M_n})_{n\in\mathbb N}$ converges weakly to $\mu$, we denote it by $\mu_{M_n} \Rightarrow \mu$, or simply $M_n \Rightarrow \mu$ as $n \to \infty$. In particular, we can find a noncommutative random variable $a$ such that $a$ has distribution $\mu$ and $M_n$ converges to $a$.

Finally, we say that $d$ random Hermitian matrices $M_1,\ldots, M_d \in M_n(\complex)$ are asymptotically free in probability if 
\begin{align*}
    \lim_{n \to \infty}\tau^{(n)}\Big(\left(M_{i_1}^{p_1}-\tau^{(n)}(M_{i_1}^{p_1})\right)\cdots 
    \left(M_{i_m}^{p_m}-\tau^{(n)}(M_{i_m}^{p_m})\right)\Big)=0
\end{align*}
in probability, for all $m \ge 1$, $i_1 \ne i_2 \ne \cdots \ne i_m \in [d]$ and $p_1,\ldots, p_m \ge 1$. They are asymptotically free in expectation if
\begin{align*}
    \lim_{n \to \infty}\esp \Big[\tau^{(n)}\Big(\left(M_{i_1}^{p_1}-\esp\left(\tau^{(n)}(M_{i_1}^{p_1})\right)\right)\cdots 
    \left(M_{i_m}^{p_m}-\esp\left(\tau^{(n)}(M_{i_m}^{p_m})\right)\right)\Big)\Big]=0,
\end{align*}
for all $m \ge 1$, $i_1\ne i_2 \ne \cdots \ne i_m \in [d]$ and $p_1,\ldots, p_m \ge 1$. This is equivalent to \eqref{freeness identity} in the limit for the algebras $\mathcal{A}(M_i)$ generated by each $M_i$. As usual, we write $M+\lambda:=M+\lambda \Id$, where $M \in M_n(\complex)$, $\lambda \in \complex$.

\begin{remark}
    We have defined weak convergence as the convergence of the moments of all order. In particular, this requires that $\tr(M_n^p)$ is integrable for all $p \ge 1$. In several cases, however, such strong integrability is not needed, and similar results can be proved for a larger class of random matrices via truncation techniques \cite[Theorem 2.1.21]{anderson2010introduction}. 
\end{remark}

\subsection{Rigorous results}
Let $\mu_1,\ldots,\mu_d$ be probability measures on $\real$ and $a_1,\ldots,a_d$ be free variables with distribution $\mu_i$, respectively. We define the tensor measure $\mu_i \star \mu_i$ as the measure associated with $a_i \otimes a_i$ and the tensor convolution $\mu_1 \circledast \cdots \circledast \mu_d$ as the distribution of
\begin{align*}
    \sum_{i \in [d]}a_i \otimes a_i=a_1 \otimes a_1 +\cdots+a_d \otimes a_d.
\end{align*}
In other words, we have
\begin{align*}
    \int x^k\, \text{d}\mu_i \star \mu_i=\tau\otimes \tau(a_i^k \otimes a_i^k)=\tau^2(a_i^k),
\end{align*}
and
\begin{align*}
    \int x^k\, \text{d}\mu_1 \circledast \cdots \circledast \mu_d=\tau\otimes \tau\left(\sum_{i \in [d]}a_i \otimes a_i\right)^k =\sum_{i \in [d]^k}\tau^2(a_{i_1}\cdots a_{i_k}).
\end{align*}
Note that $\mu_1 \circledast \cdots \circledast \mu_d$ is not necessarily the free convolution of $\mu_i \star \mu_i$. Indeed, it was proved in \cite{collins2017freeness} that freeness for tensor products $a_i\otimes a_i$ does not follow from freeness of $a_i$.

In the sequel, we will be considering random Hermitian matrices $W_1,\ldots, W_d\in M_n(\complex)$ satisfying the following assumptions: \\
\begin{enumerate}[label=(\textbf{A.\arabic*})]
    \item\label{assumption A1} For each $i \in [d]$,  $W_i\in M_n(\complex)$ converges weakly in probability and in expectation to $\mu_i$;
    \item\label{assumption A2} The family $(W_i)_{i \in [d]}$ is in probability and in expectation asymptotically free, for each $d$ fixed.
    \item\label{assumption A3} For each $i \in [d]$, $\esp (W_i \otimes \overline{W_i})\in M_{n^2}(\complex)$ converges weakly to $0$.
    \item\label{assumption A4} The matrices $(W_i)_{i \in [d]}$ are independent and identically distributed.\\
\end{enumerate}

As we will see in Section~\ref{sec: examples}, a variety of classical random matrix models satisfy those assumptions. 
We are ready to state our main result. 
\begin{theorem}\label{theorem: main theorem}
    Let $W_1,\ldots,W_d \in M_n(\complex)$ be centered Hermitian random matrices satisfying Assumptions \ref{assumption A1},\ref{assumption A2}, \ref{assumption A3} and $d=d(n)$. Then, the following holds.
    \begin{enumerate}[label=(\ref{theorem: main theorem}.\roman*)]
        \item\label{main theorem: part 1} If $d$ is fixed, then   
    \begin{align*}
        \Delta:=\frac{1}{\sqrt{d}}\sum_{i \in [d]} \left(W_i \otimes \overline{W}_i-\esp(W_i \otimes \overline{W}_i) \right)
    \end{align*}
    converges weakly in probability and in expectation to the tensor convolution of $\tilde{\mu}_1,\ldots,\tilde{\mu}_d$, where $\tilde{\mu}$ denotes the dilation of $\mu$, i.e.
    \begin{align*}
        \int x^p\, \text{d}\tilde\mu:=\frac{1}{d^{p/2}}\int x^p\, \text{d}\mu.
    \end{align*}
    This means that
    \begin{align*}
        \mu_{\Delta} \underset{n\to\infty}{\Rightarrow} \tilde\mu_1 \circledast \cdots \circledast \tilde\mu_d,
    \end{align*}
    in probability and in expectation.
    
    \item\label{main theorem: part 2} If $d=d(n)$ diverges, $(W_i)_{i \in [d]}$ are also normalized and satisfy Assumption  \ref{assumption A4}, then $\Delta$
    converges weakly in probability and in expectation to the semicircular distribution $\mu_{sc}$. This means that
    \begin{align*}
        \mu_{\Delta} \underset{n\to\infty}{\Rightarrow} \mu_{sc},
    \end{align*}
    in probability and in expectation.
    \end{enumerate}
\end{theorem}


\section{Proofs}\label{sec: proofs}
We start with the following observation.
\begin{lemma}\label{lemma: limit distribution of B_i}
Suppose that Assumptions \ref{assumption A1} and \ref{assumption A3} hold. For each $i\in[d]$, $B_i=W_i \otimes \overline{W}_i-\esp \left(W_i \otimes \overline{W}_i\right)$ converges weakly in probability and in expectation to $a_i\otimes a_i$, where $a_i$ has distribution $\mu_i$.
\end{lemma}

\begin{proof}
For Hermitian matrices $M_1,M_2 \in M_n(\complex)$, such that $M_1 \Rightarrow \mu_1$ and $M_2 \Rightarrow 0$ as $n \to \infty$, we have
\begin{align*}
    M_1+M_2 \Rightarrow \mu_1.
\end{align*}
Indeed, using Holder's Inequality, we can bound the trace of a product of the matrices $M_1$ and $M_2$ by 
\begin{align}\label{ineq: trace inequality}
    \left|\tau^{(n)}\left(\prod_{j \in [m]}M_{i_j}\right)\right| \le \prod_{j \in [m]}\left(\tau^{(n)}|M_{i_j}|^{p_j}\right)^{1/p_j},
\end{align}
where $i \in \{1,2\}^m$ and $(p_j)_{j\in[m]}$ are conjugate exponents. If there is at least one $j\in[m]$ such that $i_j=2$, then the right-hand side of \eqref{ineq: trace inequality} goes to $0$. In particular, since
\begin{align*}
    \tau^{(n)}\left((M_1+M_2)^p\right)&=\sum_{i \in \{1,2\}^p}\tau^{(n)}\left(\prod_{j \in [p]}M_{i_j}\right),
\end{align*}
any term associated with $i \in \{1,2\}^p$ having at least one index $i_j=2$ will asymptotically vanish. Therefore,  
\begin{align*}
    \lim_{n \to \infty }\tau^{(n)}\left((M_1+M_2)^p\right)=\lim_{n \to \infty }\tau^{(n)}\left(M_1^p\right).
\end{align*}
The result of the lemma follows by taking $M_1=W_i\otimes \overline{W_i}$ and $M_2=-\esp(W_i\otimes \overline{W_i})$, for each $i\in[d]$. Indeed,
\begin{align*}
    \tau^{(n)}\otimes \tau^{(n)}\left(\left(W_i \otimes \overline{W}_i\right)^p\right)&=\left(\tau^{(n)}\left(W_i^p\right)\right)^2,
\end{align*}
since the eigenvalues of $W_i$ are real. Hence, the weak convergence in probability and in expectation of $W_i$ to $\mu_i$ implies that
\begin{align*}
    \tau^{(n)}\otimes \tau^{(n)}\left(\left(W_i \otimes \overline{W}_i\right)^p\right) \underset{n\to\infty}{\to} \left(\int x^p\, \text{d}\mu_i\right)^2
\end{align*}
in probability and in expectation. This finishes the proof. 
\end{proof}

As an immediate consequence, we have the following.
\begin{corollary}\label{corollary: joint law of B_i}
    Suppose that Assumptions \ref{assumption A1},\ref{assumption A2}, and \ref{assumption A3} hold. Let $B_i=W_i\otimes \overline{W_i}-\esp(W_i \otimes \overline{W_i})$. Then, for each $d$ fixed, $(B_{i})_{i \in [d]}$ converges weakly in probability and in expectation to $(a_i\otimes a_i)_{i \in [d]}$, where $a_i$ are free variables with distribution $\mu_i$.
\end{corollary}
\begin{proof}
    By Assumption \ref{assumption A2}, we can assume $(W_i)_{i \in [d]}$ converges weakly in probability and in expectation to $(a_i)_{i \in [d]}$, where $a_i$ are free. The result follows by Lemma \ref{lemma: limit distribution of B_i}.
\end{proof}

We are ready to prove Theorem \ref{main theorem: part 1}.
\begin{proof}[Proof of Theorem \ref{main theorem: part 1}]
    By Corollary \ref{corollary: joint law of B_i}, we have
    \begin{align*}
        \sum_{i \in [d]}B_i \underset{n \to \infty}{\Rightarrow} \sum_{i \in [d]}a_i\otimes a_i,
    \end{align*}
    where $a_i$ are free with distribution $\mu_i$. To conclude, it suffices to note that $\tilde\mu_i$ is the distribution of $a_i\otimes a_i/\sqrt{d}$.
\end{proof}
Before proving Theorem \ref{main theorem: part 2}, let us give a sufficient condition for Assumption \ref{assumption A3} to hold.
\begin{lemma}\label{lemma: sufficient condition to A3}
    Let $W_1,\ldots,W_d$ satisfying Assumption \ref{assumption A1} and assume that
    \begin{align}\label{limit: centering}
        \esp(\tau(W_i)) \underset{n \to \infty}{\to}0,
    \end{align}
    for all $i \in [d]$. Let $(W_i^{(l)})_{l \in [L]}$ be independent copies of $W_i$. If for each $i \in [d]$ and $L \ge 1$, $(W_i^{(l)})_{l \in [L]}$ are asymptotically free in probability and in expectation, then Assumption \ref{assumption A3} holds.
\end{lemma}
\begin{proof}
Fix $i \in [d]$. Then, we can easily see that
\begin{align*}
   \frac{1}{n^2}\tr\left(\left(\esp \left(W_i\otimes \overline{W_i}\right)\right)^p\right)=\frac{1}{n^2}\esp\left(\tr^2\left(W_i^{(1)}\cdots W_i^{(p)}\right)\right).
\end{align*}
Since $(W_i^{(l)})_{l \in [p]}$ are asymptotically free in probability and in expectation, we get that
\begin{align*}
    \frac{1}{n^2}\tr\left(\left(\esp \left(W_i\otimes \overline{W_i}\right)\right)^p\right) \underset{n \to \infty}{\to} \tau^2(a_i^{(1)}\cdots a_i^{(p)}),
\end{align*}
where $(a_i^{(l)})_{l \in [p]}$ are free i.i.d with distribution $\mu_i$. The centered Assumption \eqref{limit: centering} implies that $\mu_i$ is centered, hence $\tau(a_i^{(1)}\cdots a_i^{(p)})=0$ by freeness and the result follows.
\end{proof}

\subsection{The asymptotic free Central Limit Theorem}\label{subsec: asymptotic free central limit theorem}
In order to prove Theorem \ref{main theorem: part 2}, we begin with a simple lemma.
\begin{lemma}\label{lemma: centering assumption}
Suppose that Assumption \ref{assumption A4} holds. For each $i\in[d]$, set $B_i:=W_i \otimes \overline{W}_i-\esp (W_i \otimes \overline{W}_i)$. Let $(i_1,\ldots,i_p)\in [d]^p$ and suppose that there exists $k\in[p]$ such that, for all $l\in[p]$ with $l\neq k$, $i_l\neq i_k$. Then, we have
\begin{align*}
    \esp \left(\tr(B_{i_1}\cdots B_{i_p})\right)=0.
\end{align*}
\end{lemma}

\begin{proof}
We can assume without loss of generality that $k=1$, i.e.~$ i_1\notin \{i_2,\ldots,i_p\}$. In this case, notice that $B_{i_1}$ is independent of $B:=B_{i_2}\cdots B_{i_p}$ by Assumption \ref{assumption A4}, hence
\begin{align*}
    \esp \left(\tr(B_{i_1}B)\right)=\sum_{k_1,k_2 \in [n]^2} \esp \left(B_{i_1}(k_1,k_2)\right)\esp \left(B(k_2,k_1)\right).
\end{align*}
Since $B_{i_1}$ is centered, we have $\esp (B_{i_1}(k_1,k_2))=0$ and the result follows.
\end{proof}

Although Lemma \ref{lemma: centering assumption} is a simple consequence of independence and centering, it is a powerful property that will allow us to obtain a limit distribution for $\Delta$. 
We are now ready to prove Theorem \ref{main theorem: part 2}.
\begin{proof}[Proof of Theorem \ref{main theorem: part 2}]
We begin by writing $\Delta$ as 
\begin{align*}
    \Delta=\frac{1}{\sqrt{d}}\sum_{i \in [d]}B_i,
\end{align*}
where $B_i=W_i \otimes \overline{W}_i-\esp (W_i \otimes \overline{W}_i)$. We recall that, given $X\in M_{n^2}(\complex)$ a random matrix, we have
\begin{align*}
    \esp \left(\tau^{(n^2)}(X)\right)=\esp\left(\frac{1}{n^2}\tr(X)\right).
\end{align*}
Then, we can compute the moments of $\Delta$ by
\begin{align*}
\esp \left(\tau^{(n^2)}(\Delta^p)\right)=d^{-p/2}\sum_{i \in [d]^p}\esp\left(\tau^{(n^2)}\left(B_{i_1}\cdots B_{i_p}\right)\right).
\end{align*}
For every $i \in [d]^p$, we associate a partition $\pi(i)$ of $[p]$ defined by placing any  $k,\ell \in [p]$ in the same block of $\pi(i)$ whenever $i_k=i_l$. We denote by $P(p)$ the set of all partitions of $[p]$. 
Notice that for every $i \in [d]^p$, $\esp\left(\tau^{(n^2)}(B_{i_1}\cdots B_{i_p})\right)$ only depends on the partition $\pi(i)$, since $B_i$ are identically distributed by Assumption \ref{assumption A4}. In particular, let $\tau^{(n^2)}(\pi)$ be the common value of $\esp\left(\tau^{(n^2)}(B_{i_1}\cdots B_{i_p})\right)$ for $\pi(i)=\pi$. Hence,
\begin{align*}
\esp\left( \tau^{(n^2)}(\Delta^p)\right)=d^{-p/2}\sum_{\pi \in P(p)}\tau^{(n^2)}(\pi)\left|\{i \in [d]^p:\pi(i)=\pi\}\right|.
\end{align*}
The cardinality can be computed by simply choosing an index for each block of $\pi$, namely,
\begin{align*}
\left|\{i \in [d]^p:\pi(i)=\pi\}\right|=d(d-1)\cdots (d-|\pi|+1)=d^{|\pi|}\left(1+O_p\left(\frac{1}{d}\right)\right),
\end{align*}
where $f=O_p(1/d)$ means that $f \le C_p/d$, for some constant $C_p$ that depends only on $p$. We then have
\begin{align*}
\esp\left( \tau^{(n^2)}(\Delta^p)\right)=\sum_{\pi \in P(p)}\tau^{(n^2)}(\pi)d^{|\pi|-p/2}\left(1+O_p\left(\frac{1}{d}\right)\right).
\end{align*}
By Lemma \ref{lemma: centering assumption}, if $\pi \in P(p)$ has a block $V$ of size $1$, then we would have $\tau^{(n^2)}(\pi)=0$. Hence, we can restrict to partitions without single blocks, yielding
\begin{align*}
\esp \left(\tau^{(n^2)}(\Delta^p)\right)=\sum_{\substack{\pi \in P(p)\\\forall V \in \pi,\,|V| \ge 2}}\tau^{(n^2)}(\pi)d^{|\pi|-p/2}\left(1+O_p\left(\frac{1}{d}\right)\right).
\end{align*}
Using Assumption \ref{assumption A1} and Holder's Inequality, we have that 
\begin{align}\label{ineq: uniform integrability}
\tau^{(n^2)}(\pi) \le \max_{i \in [d]}\left[\esp\left(\tau^{(n^2)}(|B_i|^p)\right)\right]^{1/p}\leq C_p,
\end{align}
where $C_p<\infty$ is a constant independent of $n$.
Hence, whenever $\pi$ has a block of size $|V| \ge 3$, we have $|\pi| <p/2$ and
\begin{align*}
\tau^{(n^2)}(\pi)d^{|\pi|-p/2}=O_p\left(\frac{1}{\sqrt{d}}\right).
\end{align*}
In particular, the only partitions that contribute to the dominating term are the pair partitions; that is, every block has cardinality two. Let $P_2(p)$ be the set of all pair partitions, then
\begin{align*}
\esp \left(\tau^{(n^2)}(\Delta^p)\right)=\sum_{\pi \in P_2(p)}\tau^{(n^2)}(\pi) +O_p\left(\frac{1}{\sqrt{d}}\right).
\end{align*}
Note that $\tau^{(n^2)}(\pi)$ depends only on $B_1,\ldots,B_p$. Corollary \ref{corollary: joint law of B_i} implies that $(B_i)_{i \in [d]}$ converges to $(a_i\otimes a_i)_{i \in [d]}$ in probability and in expectation, where $a_i$ are free variables and the limit of $W_i$. Hence,
\begin{align*}
\tau^{(n^2)}(\pi) \underset{n\to\infty}{\to} \tau\otimes \tau(\pi):=\tau^2(a_{i_1}\cdots a_{i_p}),
\end{align*}
In particular, freeness implies that $\tau\otimes \tau(\pi)=0$ if and only if $\pi \in P_2(p) \setminus NC_2(p)$, where $NC_2(p)$ is the set of noncrossing pair partitions. If $\pi \in NC_2(p)$, then $\tau\otimes \tau(\pi)=1$ because the matrices are normalized. We deduce that
\begin{align*}
\esp \left(\tau^{(n^2)}(\Delta^p)\right)=|NC_2(p)|+o_p(1)+O_p\left(\frac{1}{\sqrt{d}}\right), 
\end{align*}
where $f=o_p(1)$ denotes a function that depends on $p$ and $f \to 0$ as $n \to \infty$. This proves the convergence in expectation. To get convergence in probability, set
\begin{align*}
    \tau^{(n^2)}(j)=\tau^{(n^2)}(B_{j_1}\cdots B_{j_p}),
\end{align*}
for each $j \in [d]^p$. Then,
\begin{align*}
    \var\left(\tau^{(n^2)}(\Delta^p)\right)=\frac{1}{d^p}\sum_{i,j \in [d]^p}\left\{\esp \left( \tau^{(n^2)}(i) \tau^{(n^2)}(j)\right)-\esp\left( \tau^{(n^2)}(i)\right)\esp\left( \tau^{(n^2)}(j)\right)\right\}.
\end{align*}
Now, each index $i_l,j_l$ must appear at least twice in $i \cup j$, as the matrices are centered and independent. In particular, the summation has at most $d^p$ indices. We thus have
\begin{align*}
    \var\left(\tau^{(n^2)}(\Delta^p)\right) \le \max_{i,j \in [d]^p}\left\{\esp \left( \tau^{(n^2)}(i) \tau^{(n^2)}(j)\right)-\esp\left( \tau^{(n^2)}(i)\right)\esp\left( \tau^{(n^2)}(j)\right)\right\}.
\end{align*}
Assumption \ref{assumption A1} implies that $(B_i)_{i \in [p]}$ are uniformly integrable, then the Dominated Convergence Theorem \cite[Theorem 1.5.8]{durrett2019probability} and Corollary \ref{corollary: joint law of B_i} imply that
\begin{align*}
     \esp\left(\tau^{(n^2)}(i)\tau^{(n^2)}(j)\right) \underset{n\to\infty}{\to} \tau \otimes \tau(\pi(i))\tau \otimes \tau(\pi(j)),
\end{align*}
The same limit holds for $\esp\left( \tau^{(n^2)}(i)\right)\esp\left( \tau^{(n^2)}(j)\right)$, and therefore
\begin{align*}
    \var\left(\tau^{(n^2)}(\Delta^p)\right)=o_p(1),
\end{align*}
from which convergence in probability follows.
\end{proof}

\begin{remark}
As we saw in the proof, the only assumption for the existence of a limit of $\Delta$ is the centering assumption in Lemma \ref{lemma: centering assumption}. In particular, the proof of Theorem \ref{main theorem: part 2} works verbatim for noncommutative variables. If $a_1,\ldots, a_d$ are centered exchangeable variables such that Lemma \ref{lemma: centering assumption} holds, that is, $\tau(a_{i_1}\cdots a_{i_p})=0$ whenever there exists an index different than the others, we have
\begin{align*}
    \lim_{d \to \infty}\tau\left(\frac{1}{\sqrt{d}}\sum_{i \in [d]}a_i\right)^p=\sum_{\pi \in P_2(p)}\tau(\pi).
\end{align*}
If in addition $(a_i)_{i \in [d]}=(a_i^{(n)})_{i \in [d]}$ depends on some parameter $n$ and are asymptotically free, as $n$ grows, we get an asymptotic free Central Limit Theorem
\begin{align*}
    \lim_{n,d \to \infty}\tau^{(n)}\left(\frac{1}{\sqrt{d}}\sum_{i \in [d]}a_i^{(n)}\right)^p=|NC_2(p)|.
\end{align*}
\end{remark}

\section{Examples}\label{sec: examples}

We begin with an almost surely unital and trace-preserving Hermitian quantum channel.
\begin{example}[Rademacher diagonal matrices] \label{ex:Rademacher}
    Let $R=\mathrm{diag}(\varepsilon_{k})_{k \in [n]}\in M_n(\complex)$ be a Rademacher diagonal matrix, i.e.~$\varepsilon_{1},\ldots,\varepsilon_n$ are independent Rademacher (or symmetric Bernoulli) random variables with parameter $1/2$. Let $U_1,\ldots, U_d\in M_n(\complex)$ be independent Haar unitary matrices independent of $R$ as well and set $W_i=U_iRU_i^*$ for each $i\in[d]$. It follows from \cite{voiculescu1991limit} that $(W_i)_{i\in[d]}$ are asymptotically free and Lemma \ref{lemma: sufficient condition to A3} yields Assumption \ref{assumption A3}. Theorem \ref{main theorem: part 1} thus implies that 
    \begin{align*}
        \Delta \underset{n\to\infty}{\Rightarrow} \frac{1}{\sqrt{d}}\sum_{i \in [d]}r_i\otimes r_i,
    \end{align*}
    in probability and in expectation, where $(r_i)_{i\in[d]}$ are free with Rademacher distribution.  In particular, for a polynomial
    \begin{align*}
        p(x):=b_0+\sum_{l\in [m]}b_lx^l,
    \end{align*}
    we immediately compute
    \begin{align*}
        &p(r_i)-\tau(p(r_i))=\left(\sum_{l \in [\lceil m/2\rceil ]}b_{2l-1}\right)r_i;\\
        &p(r_i\otimes r_i)-\tau\otimes \tau(p(r_i\otimes r_i))=\left(\sum_{l \in [\lceil m/2\rceil]}b_{2l-1}\right)r_i\otimes r_i.
    \end{align*}
    Since $\tau(r_i)=0$, we can readily see that the freeness of $(r_i)_{i\in[d]}$ is equivalent to the freeness of $(r_i\otimes r_i)_{i\in[d]}$ in this case. Hence $(r_i\otimes r_i)_{i \in [d]}\overset{d}{=}(r_i)_{i \in [d]}$ where the equality holds in distribution. If we denote the limit of $\Delta$ by $z$, we then have
    \begin{align*}
        z\overset{d}{=}\frac{1}{\sqrt{d}}\sum_{i \in [d]}r_i.
    \end{align*}
    For $d$ even, we can precisely compute this limit. Indeed, \cite[Example 12.8.1]{nica2006lectures} shows that the free convolution $(r_1+r_2)$ has arcsine distribution. Moreover, such a distribution was proved to be the law of $u+u^*$ in \cite[Example 1.14]{nica2006lectures}, where $u$ is a Haar unitary. We then deduce that
    \begin{align*}
        z\overset{d}{=}\frac{1}{\sqrt{d}}\sum_{i \in [d/2]}(u_i+u_i^*).
    \end{align*}
    The distribution of the free sum of $u_i+u_i^*$ is called the Kesten-McKay distribution with parameter $d$; see \cite{kesten1959symmetric},\cite[Exercise 12.21]{nica2006lectures}. Its density is given by
    \begin{align*}
        f_{KM(d)}(x)=\frac{1}{2\pi}\frac{d}{d^2-x^2}\sqrt{4(d-1)-x^2}\,\mathbf{1}_{|x| \le 2\sqrt{d-1}}.
    \end{align*}
    It is also the limit spectral distribution of random $d$-regular graphs \cite{mckay1981expected}. Therefore, the density of $b$ is given by
    \begin{align*}
        \tilde f_{KM(d)}(x)=\frac{1}{2\pi}\frac{d}{d-x^2}\sqrt{4\left(1-\frac{1}{d}\right)-x^2}\,\mathbf{1}_{|x| \le 2\sqrt{1-\frac{1}{d}}}.
    \end{align*}
    A direct computation shows that, for all $x \in (-2,2)$, $\tilde f_{KM(d)}(x) \to f_{sc}(x)$ as $d \to \infty$. This is a local version of Theorem \ref{main theorem: part 2}.
\end{example}

 Let us interpret Example \ref{ex:Rademacher} above in terms of the corresponding random quantum channel. Given $W=URU^*\in M_n(\complex)$ a uniformly rotated Rademacher matrix, we have $W^2=R^2=\Id$. Hence, sampling $W_1,\ldots,W_d\in M_n(\complex)$ independently and uniformly rotated Rademacher matrices and setting $K_i=W_i/\sqrt{d}$ for each $i\in[d]$, the random completely positive map 
\[ \Phi:X\in M_n(\complex)\mapsto \sum_{i\in[d]} K_iXK_i\in M_n(\complex) \]
is exactly trace-preserving and unital (not just on average), i.e.~it is a random unital quantum channel. More precisely, it is a mixture of unitary conjugations. What is more, by what precedes, the spectral distribution of $\sqrt{d}(\Phi-\esp(\Phi))$ converges weakly almost surely to $\tilde\mu_{KM(d)}$ (for $d$ even) as $n\to\infty$. And the spectral distribution of $\sqrt{d}(\Phi-\esp(\Phi))$ converges weakly in probability and in expectation to $\mu_{sc}$ as $n,d\to\infty$.

\begin{remark}
    There is a straightforward generalization of the previous example to deterministic matrices. Let $M_1,\ldots, M_d\in M_n(\complex)$ be deterministic Hermitian matrices such that, for each $i\in[d]$, $M_i \Rightarrow \mu_i$ as $n\to\infty$ and $\mu$ is centered. A natural way to make them asymptotically free is by conjugating them with independent Haar unitaries. In particular, if we set, for each $i\in[d]$, $W_i=U_iM_iU_i^*$, where $U_1,\ldots,U_d\in M_n(\complex)$ are independent uniformly chosen unitaries, then $(W_i)_{i\in[d]}$ satisfies Assumptions \ref{assumption A1},\ref{assumption A2},\ref{assumption A3} and \ref{assumption A4}. Indeed, for instance, Assumption \ref{assumption A3} follows by Lemma \ref{lemma: sufficient condition to A3}. So that we can apply Theorem \ref{theorem: main theorem} to this kind of matrices.
\end{remark}

\begin{example}[Wigner matrices] \label{ex:Wigner}
A random Hermitian matrix $W=(W_{kl})_{k,l\in[n]} \in M_n(\complex)$ is a Wigner matrix if $W_{kl}$ are i.i.d.~centered with variance $1/n$ for $k,l\in[n]$ such that $k \leq l$ and $W_{kl}=\overline{W_{lk}}$ for $k,l\in[n]$ such that $l<k$. It is well-known \cite{anderson2010introduction,dykema1993certain,voiculescu1991limit} that independent Wigner matrices $W_1,\ldots,W_d$ satisfy Assumptions \ref{assumption A1}, \ref{assumption A2}, \ref{assumption A3} (by Lemma \ref{lemma: sufficient condition to A3}) and \ref{assumption A4}. Their joint law $(W_i)_{i \in [d]}$ converges weakly in probability to a free semicircular family $(s_i)_{i \in [d]}$. Therefore, Theorem \ref{theorem: main theorem} holds and the limit in Theorem \ref{main theorem: part 1} is expressed as
\begin{align*}
    \Delta \underset{n\to\infty}{\Rightarrow} \frac{1}{\sqrt{d}}\sum_{i \in [d]}s_i \otimes s_i,
\end{align*}
in probability and in expectation.
\end{example}

Let us interpret Example \ref{ex:Wigner} above in terms of the corresponding random quantum channel. Given $W\in M_n(\complex)$ a Wigner matrix, we have $\esp (W_{kl})=0$ and $\esp|W_{kl}|^2=1/n$ for all $k,l\in[n]$ and $(W_{kl})_{k\leq l\in[n]}$ independent. Therefore,
\[ \esp\left(W^2\right) = \sum_{k\in[n]} E_{kk} = \Id. \]
So we can conclude that, sampling $W_1,\ldots,W_d\in M_n(\complex)$ independent Wigner matrices and setting $K_i=W_i/\sqrt{d}$ for each $i\in[d]$, the random completely positive map 
\[ \Phi:X\in M_n(\complex)\mapsto \sum_{i\in[d]} K_iXK_i\in M_n(\complex) \]
is, on average, trace-preserving and unital. Moreover, the spectral distribution of $d(\Phi-\esp(\Phi))$ converges weakly almost surely to $\mu_{sc}^{\circledast d}$ as $n\to\infty$, and the spectral distribution of $\sqrt{d}(\Phi-\esp(\Phi))$ converges weakly in probability and in expectation to $\mu_{sc}$ as $n,d\to\infty$.
In this case, we can actually explicitly compute $\esp(\Phi)$. Indeed, denoting by $\{e_1,\ldots,e_n\}$ the canonical basis of $\complex^n$ and setting $E_{kl}=e_ke_l^*$ for each $k,l\in[n]$, we have
\[ \esp \left(W\otimes\overline{W}\right) = \frac{1}{n}\sum_{k,l\in[n]} E_{kl}\otimes E_{kl} + \frac{1}{n}\sum_{k\neq l\in[n]} E_{kl}\otimes E_{lk} = \psi\psi^* + \frac{1}{n}\left(F-\diag(F)\right), \]
where $\psi=\sum_{k\in[n]}(e_k\otimes e_k)/\sqrt{n}$ is a maximally entangled unit vector, $F=\sum_{k,l\in[n]} E_{kl}\otimes E_{lk}$ is the flip operator and $\diag(F)=\sum_{k\in[n]} E_{kk}\otimes E_{kk}$ is its diagonal part (with respect to the canonical product basis of $\complex^n\otimes\complex^n$). $F-\diag(F)$ has spectrum $\{1,-1,0\}$, where $1$ and $-1$ have multiplicities $n(n-1)/2$ (with associated eigenvectors $\{(e_k\otimes e_l+e_l\otimes e_k)/\sqrt{2},\ k<l\in[n]\}$ and  $\{(e_k\otimes e_l-e_l\otimes e_k)/\sqrt{2},\ k<l\in[n]\}$ respectively) and $0$ has multiplicity $n$. 
Hence,
\[ \esp(M_\Phi)= \psi\psi^* + \frac{1}{n}\left(F-\diag(F)\right). \]
This can be re-written at the level of $\Phi$ as
\[ \esp(\Phi):X\in M_n(\complex)\mapsto \tr(X)\frac{\Id}{n}+\frac{1}{n}\left(X^T-\diag(X)\right)\in M_n(\complex), \]
where $X^T$ denotes the transposition of $X$ and $\diag(X)$ its diagonal part, both with respect to the canonical basis of $\complex^n$. This means that, up to a correction that vanishes as $n$ grows, $\esp(\Phi)$ is the so-called fully randomizing channel $\Pi:X\mapsto\tr(X)\Id/n$, which has one eigenvalue equal to $1$, with associated eigenvector the maximally mixed state $\Id/n$, and all the other eigenvalues equal to $0$. Our result thus gives a precise understanding of how the asymptotic spectrum of $\Phi$ deviates from the flat one of $\Pi$, for $d$ either fixed or growing.

We can also consider non-homogeneous matrices as in \cite{lancien2023note,bandeira2023matrix}.
\begin{example}[Non-homogeneous matrices]
    Let $W_1,\ldots,W_d\in M_n(\complex)$ be independent Gaussian matrices such that, for each $i\in[d]$,
    \begin{align*}
        &\norm{\esp(W_i)} \underset{n \to \infty}{\to} 0;\\
        &\norm{\esp(W_i^2)-\text{Id}} \underset{n \to \infty}{\to} 0.
    \end{align*}
    Then it was proven in \cite[Theorem 2.10]{bandeira2023matrix} that $(W_i)_{i\in[d]}$ are asymptotically free in expectation and almost surely as long as the covariance structure of each one goes to $0$ sufficiently fast. Concretely, we have to impose that, for each $i\in[d]$,
    \begin{align*}
        v(W_i):=\norm{\cov(W_i)}=o\left(\log^{-3/2} n\right).
    \end{align*}
    The limit is a free semicircular family $(s_i)_{i \in [d]}$. So Theorem \ref{main theorem: part 1} implies that
    \begin{align*}
        \Delta \underset{n\to\infty}{\Rightarrow} \frac{1}{\sqrt{d}}\sum_{i \in [d]}s_i \otimes s_i,
    \end{align*}
    in probability and in expectation. As a matter of fact, similar results hold for Hermitian random matrices with bounded entries (see \cite[Theorem 3.25]{brailovskaya2022universality}).
\end{example}

Another classical example is Wishart-type matrices.
\begin{example}[Wishart matrices]
    A random Hermitian matrix $W \in M_n(\complex)$ is a (centered) Wishart matrix if 
    \begin{align*}
        W=XX^*-\esp (XX^*)=XX^*-\Id,
    \end{align*}
    where $X \in M_n(\complex)$ has i.i.d.~centered entries with variance $1/n$. Independent Whishart matrices $W_1,\ldots,W_d$ also satisfy Assumptions \ref{assumption A1}, \ref{assumption A2}, \ref{assumption A3} (by Lemma \ref{lemma: sufficient condition to A3}) and \ref{assumption A4};  see \cite{capitaine2007strong}. The limit of a Wishart matrix $XX^*$ is equal to $cc^*$, where $c$ is a circular element, i.e. 
    \begin{align*}
        c=\frac{s+is'}{\sqrt{2}},
    \end{align*}
    with $s,s'$ free semicircular variables, and it is called the quarter circular law or Marchenko-Pastur law. The limit law in \ref{main theorem: part 1} is given by
    \begin{align*}
        \Delta \underset{n\to\infty}{\Rightarrow} \frac{1}{\sqrt{d}}\sum_{j \in [d]}(c_jc_j^*-1) \otimes (c_jc_j^*-1),
    \end{align*}
    in probability and in expectation, where $(c_j)_{j\in[d]}$ are free copies of $c$.
\end{example}

\section{Perspectives}\label{sec: perspectives}
The results presented so far are general, and several natural classes of random matrices satisfy Assumptions \ref{assumption A1},\ref{assumption A2},\ref{assumption A3}, and \ref{assumption A4}. A natural extension would be to remove assumption \ref{assumption A3}. It turns out, however, that such an assumption is necessary in order for the limit in Theorem \ref{main theorem: part 2} to be semicircular. To understand this problem, recall that the tensor convolution $\mu_1 \circledast \mu_2$ is not the free convolution of $\mu_1 \star \mu_1$ with $\mu_2 \star \mu_2$. This is because if $a$ and $b$ are free variables, then their tensor products $a\otimes a$ and $b\otimes b$ are not necessarily free. This was proven in \cite{collins2017freeness}. In particular, it is surprising that the semicircular law can still be recovered for tensor products under only the mean zero assumption. However, a deep understanding of the structure of tensor products is still lacking, and while classical freeness is lost, tensor products still enjoy ``partial'' freeness properties. In a work in progress, we aim to decrypt such a structure in the framework of Free Probability.

A second extension of our results would be the non-Hermitian case, addressing, in particular, the conjecture mentioned in the introduction. Indeed, the methods presented here can be easily generalized to capture the convergence of the $*$-moments of non-Hermitian matrices $W_1,\ldots, W_d\in M_n(\complex)$, that is, moments of the form
\begin{align*}
    \esp\left(\tau^{(n)}(W_{i_1}^{w_1}\cdots W_{i_k}^{w_k})\right),
\end{align*}
where $i_1,\ldots,i_k \in [d]$ and $w_1,\ldots,w_k \in \{1,*\}$. However, the spectral distribution of non-Hermitian matrices is no longer characterized by the $*$-moments convergence or singular value convergence. A second ingredient, which is in general more involved, is to lower bound the smallest singular value of $W_1,\ldots, W_d$ uniformly; see \cite{bordenave2012around} and references within. The first setting we might want to study is when the Kraus operators of the random quantum channel are sampled as independent Haar unitaries. In future work, we aim to establish the quantitative invertibility of the model by deriving lower bounds for the corresponding smallest singular value.

\bibliographystyle{abbrvnat}
\bibliography{references.bib}

\end{document}